\documentclass{article}
\usepackage[utf8]{inputenc}
\usepackage{fullpage}
\usepackage{authblk}
\usepackage{amsmath,amsthm,amssymb}
\usepackage{color}

\newcommand\ignore[1]{}

\newtheorem{theorem}{Theorem}[section]
\newtheorem{definition}[theorem]{Definition}
\newtheorem{lemma}[theorem]{Lemma}
\newtheorem{proposition}[theorem]{Proposition}

\newtheorem{conjecture}[theorem]{Conjecture}

\DeclareMathOperator{\Inf}{Inf}
\DeclareMathOperator*{\EE}{\mathbb{E}}
\DeclareMathOperator{\Var}{Var}
\DeclareMathOperator{\sign}{sign}
\DeclareMathOperator{\Sym}{Sym}
\providecommand{\RR}{\mathbb{R}}
\providecommand{\NN}{\mathbb{N}}
\providecommand{\maxs}{\mathbb{S}}
\providecommand{\comp}[1]{\overline{#1}}

\def\blue{}
\newcommand*\samethanks[1][\value{footnote}]{\footnotemark[#1]}

\title{On the sum of the L1 influences of bounded functions}
\author{Yuval Filmus\thanks{Research conducted at the Simons Institute for the Theory of Computing during the 2013 fall semester on Real Analysis in Computer Science.}}
\affil{Institute for Advanced Study, Princeton, NJ}
\author{Hamed Hatami \samethanks \thanks{Research  supported in part by an NSERC, and an FQRNT grant.}}
\affil{McGill University, Montreal, QC}
\author{Nathan Keller\thanks{Research supported in part by I.S.F. grant 402/13 and by the
Alon fellowship.}}
\author{Noam Lifshitz}
\affil{Bar Ilan University, Israel}

\begin{document}

\maketitle

\begin{abstract}
Let $f\colon \{-1,1\}^n \to [-1,1]$ have degree $d$ as a multilinear polynomial. It is well-known that the total influence of $f$ is at most $d$. Aaronson and Ambainis asked whether the total $L_1$ influence of $f$ can also be bounded as a function of $d$. Ba\v{c}kurs and Bavarian answered this question in the affirmative, providing a bound of $O(d^3)$ for general functions and $O(d^2)$ for homogeneous functions. We improve on their results by providing a bound of $d^2$ for general functions and $O(d\log d)$ for homogeneous functions. In addition, we prove a bound of $d/(2 \pi)+o(d)$ for monotone functions, and provide a matching example.
\end{abstract}

\section{Introduction} \label{sec:introduction}

Let $f\colon \{-1,1\}^n \to \{-1,1\}$ be a Boolean function. The \emph{influence} of the $i$th variable is
\[
 \Inf_i[f] = \Pr_{x \sim \{-1,1\}^n} [f(x) \neq f(x \oplus e_i)],
\]
where $x \oplus e_i$ is obtained from $x$ by flipping the $i$th coordinate. The \emph{total influence} of the function is
\[
 \Inf[f] = \sum_{i=1}^n \Inf_i[f].
\]
We define $\deg f$ as the degree of the unique multilinear polynomial representing $f$. It is well-known that $\Inf[f] \leq \deg f$, and much of the usefulness of influence in the study of Boolean functions rests on this property.

The notion of influence can be extended in several ways to real-valued functions $f\colon \{-1,1\}^n \to \RR$. For each $p > 0$, one can define
\[
 \Inf^{(p)}_i[f] = \EE_{x \sim \{-1,1\}^n} \left[ \left|\frac{f(x)-f(x\oplus e_i)}{2}\right|^p \right], \quad \Inf^{(p)}[f] = \sum_{i=1}^n \Inf^{(p)}_i[f].
\]
When $f$ is Boolean, all these definitions agree with the original definition. It is well-known that $\Inf^{(2)}[f] \leq \deg f \cdot \Var[f] \leq \deg f \cdot \|f\|_\infty^2$. While studying the query complexity of partial functions, Aaronson and Ambainis~\cite{AA} asked whether $\Inf^{(1)}[f]$ can be bounded similarly. In other words, does every $f\colon \{-1,1\}^n \to [-1,1]$ of degree $d$ satisfy $\Inf^{(1)}[f] = O(d^{O(1)})$?

Ba\v{c}kurs and Bavarian~\cite{BB} answered this in the affirmative, showing that $\Inf^{(1)}[f] = O(d^3)$. When $f$ is homogeneous (that is, the unique multilinear polynomial representing $f$ is homogeneous), they obtain an improved bound $\Inf^{(1)}[f] = O(d^2)$.

\paragraph{Our results} Our main result is the bound {\blue $\Inf^{(p)}[f] \leq d^{3-p}$ for $1 \leq p \leq 2$, which implies (and follows from)} $\Inf^{(1)}[f] \leq d^2$. When $f$ is homogeneous, we are able to show that $\Inf^{(1)}[f] = O(d\log d)$. 
When $f$ is symmetric and $d \ll n^{1/2}$, we show that $\Inf^{(1)}[f] \leq d + o(d)$. Following Ba\v{c}kurs and Bavarian, we conjecture that the bound $\Inf^{(1)}[f] \leq d$ holds for \emph{all} functions $f\colon \{-1,1\}^n \to [-1,1]$ of degree $d$.

When $f$ is monotone, we show that $\Inf^{(1)}[f] \leq d/ (2\pi)+o(d)$ and provide a matching example, based on combinations of Jacobi polynomials. Note that even in the special case where $f$ is further assumed to be Boolean, our result improves the previously known bound of $\Inf[f] \leq \ln 2 \cdot \deg[f] (1+o(1))$ due to Scheder and Tan~\cite{ST13}. In this case a strong bound of $\Inf[f]  \le \sqrt{d}$ is conjectured by Gopalan and Servedio (See Conjecture~\ref{conj:GopalanServedio} below).

\paragraph{Background and applications} As mentioned above, the question $\Inf^{(1)}[f] \stackrel{?}{=} \deg^{O(1)} f \cdot \|f\|_\infty$ first appears in a paper of Aaronson and Ambainis~\cite{AA} which studies situations in which quantum algorithms can only be polynomially faster than classical algorithms. One conjecture they are interested in states that any problem with quantum query complexity $T$ can be approximately solved on most inputs by a classical algorithm that makes $T^{O(1)}$ queries. While unable to prove the conjecture, Aaronson and Ambainis reduce it to a conjecture on bounded polynomials, known henceforth as the Aaronson--Ambainis conjecture, which states that a degree $d$ polynomial $f$ satisfying $0 \leq f \leq 1$ on the cube $\{0,1\}^N$ has a variable whose influence is at least $\Inf^{(1)}_i[f] \geq (\Var^{(1)}[f]/d)^{O(1)}$, where $\Var^{(p)}[f] = \EE[|f-\EE f|^p]$. The original version of the reduction made implicit use of the bound $\Inf^{(1)}[f] = \deg^{O(1)} f \cdot \|f\|_\infty$,
as noticed by Ba\v{c}kurs. Prompted by this, Aaronson and Ambainis updated their paper~\cite{AA14} to use $\Inf^{(2)},\Var^{(2)}$ instead of $\Inf^{(1)},\Var^{(1)}$ (so that they could use the known bound $\Inf^{(2)}[f] \leq \deg f \cdot \|f\|_\infty^2$), and also showed that both formulations of their conjecture are equivalent. Separately, Ba\v{c}kurs and Bavarian~\cite{BB} managed to prove $\Inf^{(1)}[f] = O(\deg^3 f \cdot \|f\|_\infty)$, thus salvaging the original proof of Aaronson and Ambainis.

As an application of their result, Ba\v{c}kurs and Bavarian provide a simple proof of a theorem of Erd\H{o}s et al.~\cite{EGPS} regarding cuts in graphs. The theorem states that that a graph $G=(V,E)$ on $n$ vertices with density $\rho = |E|/\binom{n}{2}$ always has a cut $(S,\overline{S})$ satisfying $|E(S,\overline{S}) - \rho| = \Omega(\min(\rho,1-\rho)n^{3/2})$. The proof uses the bound $\Inf^{(1)}[f] = O(\deg^3 f \cdot \|f\|_\infty)$ for a quadratic polynomial $f$. Since the degree is constant, our improved bound only translates to an improved hidden constant in the statement of the Erd\H{o}s et al.\ result; indeed, the result is tight up to a constant, for example for random graphs.

{\blue Finally we make a simple observation that might be interesting to some readers: The bound on $\Inf^{(1)}[f]$ implies that if a function $f\colon \{-1,1\}^n \to [-1,1]$ of degree $d$ is invariant under some transitive group action then $\Var[f] \leq \frac{e^{O(d)}}{n}$.
This improves on the bound $\Var[f] \leq \frac{e^{O(d)}}{\sqrt{n}}$ that follows from a result of Dinur~et~al.~\cite{DFKO}.}

\paragraph{Paper organization} Section~\ref{sec:definitions} defines various notations used in the paper. Section~\ref{sec:upper-bounds} contains our upper bounds {\blue and an application}. Section~\ref{sec:tightness} describes several functions for which the conjectured bound $\Inf^{(1)}[f] \leq d$ is tight or almost tight. Section~\ref{sec:future} contains several conjectures which would result in improvements to our main theorems. We believe that these conjectures are interesting in their own right.

\paragraph{Acknowledgements} We thank Mohammad Bavarian, Guy Kindler, Oleksiy Klurman, Elchanan Mossel and Krzysztof Oleszkiewicz for helpful discussions.
This material is based upon work supported by the National Science Foundation under agreement No.~DMS-1128155. Any opinions, findings and conclusions or recommendations expressed in this material are those of the authors and do not necessarily reflect the views of the National Science Foundation.

\section{Definitions} \label{sec:definitions}

We use the notation $[n] = \{1,\ldots,n\}$. The complement of a set $S \subseteq [n]$ will be denoted $\comp{S} = [n] \setminus S$. Probabilities or expectations over $\{-1,1\}^n$ are always with respect to the uniform probability measure. The point $(1,\ldots,1) \in \{-1,1\}^n$ will be denoted $\mathbf{1}$. A point $(x_1,\ldots,x_n) \in \RR^n$ will be abbreviated by $\mathbf{x}$.

\paragraph{Functions} In this paper we consider functions $f\colon \{-1,1\}^n \to \RR$. A function $f$ is \emph{Boolean} if $f$ only attains the values $\pm 1$. We think of a function $f\colon \{-1,1\}^n \to \RR$ as having $n$ input variables $x_1,\ldots,x_n$ which are $\pm 1$-valued. Every such function has a unique expansion as a multilinear polynomial over the variables $x_1,\ldots,x_n$; this expansion is known as the \emph{Fourier expansion} of $f$. Each set $S \subseteq [n]$ corresponds to a multilinear monomial $\chi_S = \prod_{i \in S} x_i$ known as a \emph{Fourier character} or a \emph{Walsh function}. The coefficient of $\chi_S$ in the expansion of $f$ is known as the \emph{Fourier coefficient} $\hat{f}(S)$.

The \emph{degree} of $f$, denoted by $\deg f$, is the degree of its Fourier expansion. If all monomials appearing in the Fourier expansion of $f$ have the same degree, then $f$ is \emph{homogeneous}. If $f(\mathbf{x})$ depends only on $x_1+\cdots+x_n$ then $f$ is \emph{symmetric}. If for any $x,y$
such that $x_i \leq y_i$ for all $i$, we have $f(x) \leq f(y)$, then $f$ is \emph{monotone increasing}.

\paragraph{Symmetrization} For $f\colon \{-1,1\}^n \to \RR$, we define the \emph{symmetrization} of $f$ as
\[
\Sym(f)(x_1,x_2,\ldots,x_n) = \EE_{\sigma \in S_n} [f(x_{\sigma(1)},x_{\sigma(2)},\ldots,x_{\sigma(n)})].
\]
Similarly, for any $m \geq n$, the $m$-coordinate symmetrization of $f$, $\Sym_m(f)\colon\{-1,1\}^m \rightarrow \mathbb{R}$, is the symmetrization of the
function $\tilde{f}\colon\{-1,1\}^m \rightarrow \mathbb{R}$ defined as $\tilde{f}(x_1,\ldots,x_m)=f(x_1,\ldots,x_n)$.
As $\Sym_m(f)$ is obtained from $f$ by averaging, it is clear that $\|\Sym_m(f)\|_\infty \leq \|f\|_\infty$, and that $\deg \Sym_m f \leq \deg f$.

\paragraph{Influence} For $x \in \{-1,1\}^n$, we define $x \oplus e_i$ as the vector obtained from $x$ by flipping the $i$th coordinate. For a function $f\colon \{-1,1\}^n \to \RR$ and $i \in [n]$, we define
\[ f_i(\mathbf{x}) = \frac{f(\mathbf{x}) - f(\mathbf{x}\oplus e_i)}{2} = \sum_{\substack{S \subseteq [n] \\ i \in S}} \hat{f}(S) \chi_S = x_i \frac{\partial f}{\partial x_i} (\mathbf{x}). \]
The $i$th \emph{influence} of $f$ is $\Inf_i[f] = \|f_i\|_1$ (in the introduction, we denoted this quantity by $\Inf^{(1)}_i[f]$, but for brevity we remove the superscript in the rest of the paper). The \emph{total influence} of $f$ is $\Inf[f] = \sum_{i=1}^n \Inf_i[f]$. Alternatively, if we define
\[ \Delta(f)(\mathbf{x}) = \sum_{i=1}^n |f_i(\mathbf{x})|, \]
then $\Inf[f] = \|\Delta(f)\|_1$. When $f$ is Boolean, $\Delta(f)(\mathbf{x})$ is the \emph{sensitivity} of $f$ at $\mathbf{x}$, which is the number of indices $i \in [n]$ such that $f(\mathbf{x} \oplus e_i) \neq f(\mathbf{x})$. The quantity $\Inf[f]$ is also known as the \emph{average sensitivity} of $f$, and $\maxs(f) = \|\Delta(f)\|_\infty$ is also known as the \emph{maximum sensitivity} of $f$.

\paragraph{Noise} For $f\colon \{-1,1\}^n \to \RR$ and $x \in \RR$, the noise operator $T_\rho$ takes the function $f$ to the function $T_\rho f$ given by
\[ T_\rho f = \sum_{S \subseteq [n]} \rho^{|S|} \hat{f}(S) \chi_S. \]
When $|\rho| \leq 1$, the noise operator has the following alternative interpretation. Fix a point $\mathbf{x} \in \{-1,1\}^n$. For each $i \in [n]$, independently let $y_i$ be the unique $\pm 1$-valued random variable such that $\EE[x_i y_i] = \rho$. Then
\[ T_\rho f(\mathbf{x}) = \EE_{\mathbf{y}} [f(\mathbf{y})]. \]

\paragraph{Chebyshev polynomials} For each $d \geq 0$, the Chebyshev polynomial (of the first kind) $T_d$ is the unique univariate polynomial such that $T_d(\cos \theta) = \cos (d\theta)$. The polynomial $T_d$ has degree $d$, and is given by the recurrence $T_{d+1}(x) = 2xT_d(x) - T_{d-1}(x)$ with base cases $T_0(x) = 1$ and $T_1(x) = x$. 

\paragraph{Jacobi polynomials} The Jacobi polynomials $J_d^{\alpha,\beta}$ are a family of polynomials that are orthogonal with respect to the weight function $(1-x)^\alpha (1+x)^\beta$ on $[-1,1]$. For all $\alpha,\beta>-1$, $J_d^{\alpha,\beta}\colon [-1,1] \to \RR$ is a degree $d$ polynomial given by
\[
J_d^{\alpha,\beta}(x) = 2^{-d} \sum_{j=0}^d {{d+\alpha}\choose{j}} {{d+\beta}\choose{d-j}} (x-1)^{d-j} (x+1)^j.
\]
The Chebychev polynomial $T_d$ is equal (up to normalization) to the Jacobi polynomial $J_d^{-1/2,-1/2}$.

\section{Upper bounds} \label{sec:upper-bounds}

In this section we assume that $f\colon \{-1,1\}^n \to [-1,1]$ has degree $d$. We prove the following upper bounds on the total influence:
\begin{enumerate}
 \item $\Inf[f] \leq d^2${\blue, and more generally $\Inf^{(p)}[f] \leq d^{3-p}$ for $1 \leq p \leq 2$}.
 \item If $f$ is homogeneous then $\Inf[f] = O(d\log d)$.
 \item If $f$ is symmetric and $d \ll n^{1/2}$ then $\Inf[f] \leq d + o(d)$.
 \item If $f$ is monotone then $\Inf[f] \leq d/(2\pi) + o(d)$.
\end{enumerate}

{\blue As an application, we prove that if $f$ is invariant under some transitive group action then $\Var[f] \leq \Inf^{(2)}[f] \leq \frac{e^{O(d)}}{n}$.}

\subsection{Upper bound for general functions}
\label{sec:sub:general-bound}

The upper bound $d^2$ for general functions uses a Bernstein--Markov type inequality.   The classical Bernstein--Markov theorem provides an upper bound on the derivative of a polynomial that is bounded in an interval.

\begin{proposition}[Bernstein--Markov] \label{pro:bernstein}
Let $p\colon [-1,1] \to \mathbb{R}$ be a polynomial of degree $d$. For every $x \in [-1,1]$,
\[ p'(x) \leq \min \left(d^2, \frac{d}{\sqrt{1-x^2}}\right) \|p\|_\infty. \]
\end{proposition}

The generalization that we will use, due to  Sarantopoulos~\cite{S}, extends Proposition~\ref{pro:bernstein} to Banach spaces.
Using the classical Bernstein--Markov theorem instead results in the slightly weaker upper bound $2d^2$.


Sarantopoulos's theorem concerns polynomials in general Banach spaces. Since in this paper we only need the finite dimensional case, to avoid introducing unnecessary terminology, we will state Sarantopoulos's theorem for the special case of finite dimensional Banach spaces. Recall that for a \emph{finite dimensional} Banach space $E=(\mathbb{R}^n,\|\cdot\|)$, the Fr\'echet derivative of a differentiable function $f\colon E \to \mathbb{R}$ at a point $\mathbf{x}$ is the linear operator $Df(\mathbf{x})\colon E \to \mathbb{R}$ defined as
\[ Df(\mathbf{x}) \colon \mathbf{y} \mapsto \sum_{i=1}^n y_i \frac{\partial f}{\partial x_i}(\mathbf{x}). \]
\begin{proposition}[{Sarantopoulos~\cite[Theorems 1 and 2]{S}}] \label{pro:sar}
 Let $E=(\mathbb{R}^n,\|\cdot\|)$ be a finite dimensional Banach space and $P\colon \mathbb{R}^n \to \RR$ be a  polynomial of degree $d$ satisfying $|P(\mathbf{x})| \leq 1$ for all $\|\mathbf{x}\| \leq 1$. Then $|DP(\mathbf{x})\mathbf{y}| \leq \min\left(d^2,\frac{d}{\sqrt{1-\|x\|^2}}\right)$ for all $\|\mathbf{x}\|,\|\mathbf{y}\| \leq 1$, where $DP$ is the Fr\'echet derivative of $P$.
\end{proposition}

\begin{theorem} \label{thm:d2}
 Let $f\colon \{-1,1\}^n \to [-1,1]$ be a function of degree $d$. Then
 \[ \Inf[f] \leq \|\Delta(f)\|_\infty \leq d^2. \]
\end{theorem}
\begin{proof}
 Clearly $\Inf[f] = \|\Delta(f)\|_1 \leq \|\Delta(f)\|_\infty$, and so it suffices to show that $|\Delta(f)(\mathbf{x})| \leq d^2$ for all $\mathbf{x} \in \{-1,1\}^n$.
 Consider now $[-1,1]^n$ as the unit ball in the Banach space $(\RR^n, \|\mathord{\cdot}\|_\infty)$. The Fr\'echet derivative of $f$ at the point $\mathbf{x}$ is the linear operator $Df(\mathbf{x})$ given by
\[ Df(\mathbf{x})\mathbf{y} = \sum_{i=1}^n y_i \frac{\partial f}{\partial x_i}(\mathbf{x}) = \sum_{i=1}^n y_i \frac{f_i(\mathbf{x})}{x_i}. \]
 In particular, for every $\mathbf{x} \in \{-1,1\}^n$, there is some $\mathbf{y} \in \{-1,1\}^n$ such that
\[ \Delta(f)(\mathbf{x}) = Df(\mathbf{x})\mathbf{y}. \]
 Proposition~\ref{pro:sar} immediately implies that $|\Delta(f)(\mathbf{x})| \leq d^2$ for all $\mathbf{x} \in \{-1,1\}^n$.
\end{proof}

The argument in fact gives a bound on $\|\Delta(f)\|_\infty$, and in this respect, it is tight. Indeed, consider the functions $f_n(x_1,\ldots,x_n) = T_d(\frac{x_1+\cdots+x_n}{n})$. At the point $\mathbf{1}$ we have
\[
 \lim_{n \to \infty} \Delta(f_n)(\mathbf{1}) = \lim_{n \to \infty} n \left|\frac{T_d(1) - T_d(1-\tfrac{2}{n})}{2}\right|= T'_d(1) = d^2.
\]

{\blue
A simple application of H\"older's inequality allows us to interpolate between the bounds $\Inf^{(1)}[f] \leq d^2$ and $\Inf^{(2)}[f] \leq d$.

\begin{proposition}
Let $f\colon \{-1,1\}^n \to [-1,1]$ be a function of degree~$d$, and let $1<p<2$. Then $\Inf^{(p)}[f]\leq d^{3-p}$.
\end{proposition}
\begin{proof}
By H\"{o}lder's inequality, applied with the conjugate norms $q=\frac{1}{2-p}, q'=\frac{1}{p-1}$, we have
\begin{align*}
\Inf^{(p)}[f]=\sum_{i\in[n]}\EE[|f_i|^{p}] & =\sum_{i\in[n]}\EE[|f_i|^{2-p}|f_i|^{2p-2}]\leq\sum_{i\in[n]}\left(\Inf_{i}^{(1)}[f]\right)^{2-p}
\left(\Inf_{i}^{(2)}[f]\right)^{p-1}.
\end{align*}
Applying  H\"{o}lder's inequality with the same norms, but now to the outer sum,  we get
\[
\sum_{i\in[n]}\left(\Inf_{i}^{(1)}[f]\right)^{2-p}
\left(\Inf_{i}^{(2)}[f]\right)^{p-1} \leq\left(\Inf^{(1)}[f]\right)^{2-p}\left(\Inf^{(2)}[f]\right)^{p-1}\leq d^{3-p}.
\]
This completes the proof.
\end{proof}

For $p \ge 2$ we obviously have $\Inf^{(p)}[f] \le \Inf^{(2)}[f] \le d$, which is sharp as Fourier characters of degree $d$ demonstrate.
}

\subsection{Upper bound for homogeneous functions}
\label{sec:sub:homogeneous-bound}

The upper bound $O(d\log d)$ for homogeneous functions uses a result of Harris~\cite{H2}.

\begin{proposition}[{Harris~\cite{H2}}] \label{pro:har}
 Let $h$ be a real polynomial satisfying $|h(\epsilon)| \leq (1+|\epsilon|)^d$ for all $\epsilon \in \RR$. Then $|h'(0)| = O(d\log d)$.
\end{proposition}

We comment that R\'ev\'esz and Sarantopoulos~\cite{RS} show that the bound $O(d\log d)$ is optimal.

\begin{theorem} \label{thm:dlogd}
 Let $f\colon \{-1,1\}^n \to [-1,1]$ be a homogeneous function of degree $d$. Then
 \[ \Inf[f] \leq \|\Delta(f)\|_\infty \leq O(d\log d). \]
\end{theorem}
\begin{proof}
 Since $\Inf[f] = \|\Delta(f)\|_1 \leq \|\Delta(f)\|_\infty$, it suffices to show that $|\Delta(f)(\mathbf{1})| \leq O(d\log d)$.
 Let $S$ be the set of $i \in [n]$ such that $f_i(\mathbf{1}) \geq 0$. Then
 \[
  |\Delta(f)(\mathbf{1})| = \sum_{i \in S} f_i(\mathbf{1}) - \sum_{i \in \comp{S}} f_i(\mathbf{1}).
 \]
 Define the bivariate polynomial $g(x,y) = f(\overbrace{x,\ldots,x}^{S},\overbrace{y,\ldots,y}^{\comp{S}})$. Since $f$ is multilinear, its extension to the continuous cube $[-1,1]^n$ is also bounded in absolute value by $1$. This, together with homogeneity of $f$, implies $|g(x,y)| \leq \max(|x|,|y|)^d$. In particular, the function $h(\epsilon) = g(1+\epsilon,1-\epsilon)$ is a polynomial satisfying $|h(\epsilon)| \leq \max(|1+\epsilon|,|1-\epsilon|)^d = (1+|\epsilon|)^d$. Proposition~\ref{pro:har} implies that $|h'(0)| = O(d\log d)$.  Now the theorem follows as
 \[
h'(0) = \sum_{i \in S} \frac{\partial f}{\partial x_i} \frac{d(1+\epsilon)}{d\epsilon}(0) + \sum_{i \in \bar{S}}
\frac{\partial f}{\partial x_i} \frac{d(1-\epsilon)}{d\epsilon}(0) =  \sum_{i \in S} f_i(\textbf{1}) - \sum_{i \in \bar{S}} f_i(\textbf{1}) = |\Delta(f)(\textbf{1})|. \qedhere
\]
\end{proof}

In Section~\ref{sec:future} we discuss a variant of this argument which could result in better bounds.

\medskip

\noindent When $f$ is not only homogeneous but also Boolean, we can determine both $\Inf[f]$ and
$\|\Delta(f)\|_\infty$ exactly.
\begin{proposition}\label{Prop:Hom-Bool}
Let $f\colon\{-1,1\}^n \rightarrow \{-1,1\}$ be a homogeneous Boolean function of degree $d$.
Then for any $x \in \{-1,1\}^n$, $\Delta(f)(x)=d$. In particular, $\Inf[f] =
\|\Delta(f)\|_\infty = d$.
\end{proposition}

The simplest example of a homogeneous Boolean function is a Fourier character. Other examples are
discussed in Section~\ref{sec:tightness}.

\begin{proof}
Since $f$ is Boolean, for any $x$ and for any $i$, we have $f(x)-f(x \oplus e_i)
\in \{2f(x),0\}$. Thus, for a fixed $x$, all terms of the form $(f(x)-f(x \oplus e_i))/2$
have the same sign. Hence,
\begin{align*}
\Delta(f)(x) &= \sum_{i=1}^n \left| \frac{f(x)-f(x \oplus e_i)}{2} \right| = \left| \sum_{i=1}^n
\frac{f(x)-f(x \oplus e_i)}{2} \right| = \left| \sum_{i=1}^n \sum_{\{S:i \in S\}} \hat f(S) \chi_S(x) \right| \\
&= \left| \sum_{S \subset [n]} \sum_{i \in S} \hat f(S) \chi_S(x) \right| =
\left| \sum_{S \subset [n]} d \hat f(S) \chi_S(x) \right| = |df(x)|=d,
\end{align*}
the second to last equality using the homogeneity of $f$ and the last equality using the
Booleanity of $f$.
\end{proof}

We note that for bounded functions, the same proof can be applied to the local
extremum points, that is, to any $x_0$ such that either $f(x_0) \geq f(x_0 \oplus e_i)$
for all $i$ or $f(x_0) \leq f(x_0 \oplus e_i)$ for all $i$. For such points, the argument
implies $\Delta(f)(x_0) = d |f(x_0)|$.

The bound $\|\Delta(f)\|_\infty \le d$ of Proposition~\ref{Prop:Hom-Bool} does not necessarily hold for non-Boolean functions. Indeed, consider the function $f\colon \{-1,1\}^{2n} \to [-1,1]$ defined as
\[
f(x_1,\ldots,x_{2n}) = \left(\frac{\sum_{i=1}^n x_i}{n} \right)^2 - \left(\frac{\sum_{i=n+1}^{2n} x_i}{n} \right)^2 =
\frac{2}{n^2}\left(\sum_{1 \le i <j \le n} x_i x_j -\sum_{n <i <j \le 2n} x_i x_j \right).
\]
This is a homogeneous polynomial of degree $2$, but
\[ \Delta(f)(\mathbf{1})= 2n \times \frac{2}{n^2}(n-1) = 4\left(1-\frac{1}{n}\right), \]
which can be made arbitrarily close to $4$ by taking $n$ to be sufficiently large.

\subsection{Upper bounds for symmetric functions}
\label{sec:sub:symmetric}

We present two upper bounds for symmetric functions: a bound of $d+O(\frac{d^3+d\log dn}{n})$ for $d \ll n^{1/2}$ and a stronger bound of $d+O(\sqrt{dn}\exp(-n/d^4))$ for $d \ll n^{1/4}$. Both bounds use the classical Bernstein--Markov theorem on real polynomials (Proposition~\ref{pro:bernstein} above).

\begin{lemma}\label{lemma:symmetric-new}
Let $f\colon \left\{ -1,1\right\} ^{n}\rightarrow \left[-1,1\right]$ be a symmetric
function of degree $d$, where $n > d^{2}$. Then we can write
$f(\mathbf{x}) = p(\frac{x_1+\cdots+x_n}{n})$ for some polynomial
$p\colon [-1,1] \to \mathbb{R}$ of degree $d$, such that $\|p\|_\infty \leq \frac{n}{n-d^2}$.
\end{lemma}
\begin{proof}
It is easy to see that $f$ can be written as $f(\mathbf{x}) = p(\frac{x_1+\cdots+x_n}{n})$
for a unique polynomial $p\colon [-1,1] \to \mathbb{R}$ of degree $d$. So we only need to find an
upper bound on $\|p\|_{\infty}$. Suppose $x \in \left[-1,1\right]$ is such that
$\left|p\left(x\right)\right|=\|p\|_{\infty}$. Choose $y=(-1+2i/n)$  with $i \in \mathbb{Z}$ such that $|x-y|$ is minimal. Clearly, $\left|x-y\right|\leq\frac{1}{n}$, and $|p(y)| \leq 1$ since $p$ agrees
with $f$ on $y$. By the Mean Value theorem, $\frac{p(x)-p(y)}{x-y} = p'(z)$ for some $z$ between $x$ and
$y$. Thus, by Proposition~\ref{pro:bernstein},
\[
\|p\|_{\infty}=|p(x)|\leq |p'(z)||y-x|+|p(y)|
\leq\frac{d^{2}\|p\|_{\infty}}{n}+1.
\]
The assertion follows.
\end{proof}

\begin{theorem} \label{thm:dsymmetric}
 Let $f\colon \{-1,1\}^n \to [-1,1]$ be a symmetric function of degree $d$, where $d < \sqrt{n}$. Then
 \[
 \Inf[f] \leq \frac{n}{n-d^2} \left(d + O\left(\frac{d\log (dn)}{n}\right) \right).
 \]
\end{theorem}
\begin{proof}
 Theorem~\ref{thm:d2} allows us to assume that $d \geq 2$.
 By Lemma~\ref{lemma:symmetric-new}, we can write $f(\mathbf{x}) = p(\frac{x_1+\cdots+x_n}{n})$ for some polynomial
 $p\colon [-1,1] \to \mathbb{R}$ of degree $d$ with $\|p\|_\infty \leq \frac{n}{n-d^2}$. We can calculate explicitly
\begin{equation}\label{Eq:New-Symmetric1}
 \Inf[f] = n\Inf_n[f] = n\EE_{\mathbf{x} \in \{-1,1\}^{n-1}}\left[ \left|\frac{p(\tfrac{S+1}{n}) - p(\tfrac{S-1}{n})}{2}\right| \right], \text{ where } S = x_1 + \cdots + x_{n-1}.
\end{equation}
 The Mean Value theorem shows that for some $\theta_S \in [-1,1]$,
\[
 \frac{1}{2}|p(\tfrac{S+1}{n}) - p(\tfrac{S-1}{n})| = \frac{1}{n} |p'(\tfrac{S+\theta_S}{n})| \leq \frac{1}{n} \min \left(d^2, \frac{d}{\sqrt{1-(|S|+1)^2/n^2}}\right) \cdot \frac{n}{n-d^2},
\]
 using Proposition~\ref{pro:bernstein}. Let $T = \sqrt{n\log (dn)}$. Then
\begin{align*}
 \Inf[f] &= n\EE_{\mathbf{x} \in \{-1,1\}^{n-1}}\left[ \left|\frac{p(\tfrac{S+1}{n}) - p(\tfrac{S-1}{n})}{2}\right| \right] \\ &\leq
 \left(\frac{d}{\sqrt{1-(T+1)^2/n^2}} + d^2 \Pr[|S| > T] \right) \cdot \frac{n}{n-d^2} \\ &\leq
 \left( \frac{d}{\sqrt{1-O(\log (dn)/n)}} + 2d^2 e^{-2T^2/n} \right) \cdot \frac{n}{n-d^2} \\ &\leq \left(d + O\left(\frac{d\log (dn)}{n}\right) + \frac{2}{n^2} \right) \cdot \frac{n}{n-d^2} = \left( d + O\left(\frac{d \log (dn)}{n}\right) \right)\cdot \frac{n}{n-d^2},
\end{align*}
 using Hoeffding's bound in the second inequality.
\end{proof}

In the following, we prove a stronger bound, effective for $d \ll n^{1/4}$. We need two lemmas,
that may be of independent interest. The first lemma bounds the sum of first-level Fourier coefficients
of low-degree bounded functions.
\begin{lemma}
Let $f\colon \left\{ -1,1\right\} ^{n}\rightarrow \left[-1,1\right]$ be a function
of degree $d$. Then
\[
\mathbb{E}\left[\left(x_{1}+x_{2}+\cdots+x_{n}\right)f\left(x_{1},x_{2},\ldots,x_{n}\right)\right]\leq d.
\]
\end{lemma}

We note that the same result for Boolean functions is trivial, as for any Boolean $f$ of degree $d$, we have $\mathbb{E}\left[\left(x_{1}+x_{2}+\cdots+x_{n}\right)f\left(x_{1},x_{2},\ldots,x_{n}\right)\right]
\leq \Inf[f] \leq d$.

\begin{proof}
Let $f$ be as in the assumption. For any $m> \max(n,d^{2})$ we have
\[
\mathbb{E}\left[\left(x_{1}+x_{2}+\cdots+x_{n}\right)f\left(x_{1},x_{2},\ldots,x_{n}\right)\right]
=\mathbb{E}\left[\left(x_{1}+x_{2}+\cdots+x_{m}\right)f\left(x_{1},x_{2},\ldots,x_{n}\right)\right].
\]
Let $g$ be the $m$-coordinate symmetrization of $f$. It is easy to see that
\[
\mathbb{E}\left[\left(x_{1}+x_{2}+\cdots+x_{m}\right)g\left(x_{1},x_{2},\ldots,x_{m}\right)\right]
=\mathbb{E}\left[\left(x_{1}+x_{2}+\cdots+x_{m}\right)f\left(x_{1},x_{2},\ldots,x_{n}\right)\right].
\]
Since $g$ is symmetric, by Theorem~\ref{thm:dsymmetric} we have
\[
\mathbb{E}\left[\left(x_{1}+x_{2}+\cdots+x_{m}\right)g\left(x_{1},x_{2},\ldots,x_{m}\right)\right]
\leq \Inf\left[g\right] \leq \frac{m}{m-d^{2}}\left(d+O\left(\frac{d\log\left(dm\right)}{m}\right)\right).
\]
The assertion follows by tending $m$ to infinity.
\end{proof}

The next lemma shows an improved upper bound on the influence of bounded symmetric
functions that satisfy a certain monotonicity condition.
\begin{lemma}
Let $n\in\mathbb{N}$ and let $p\colon[-1,1] \to \mathbb{R}$ be a polynomial of degree $d$ that is monotone
in the interval $\left[-\frac{t}{\sqrt{n}}-\frac{2}{n},\frac{t}{\sqrt{n}}+\frac{2}{n}\right]$.
Define $f\colon\{-1,1\}^n \rightarrow \mathbb{R}$ by $f(x_{1},\ldots,x_{n}) =
p(\frac{x_{1}+\cdots+x_{n}}{n})$. If $|f(x)| \leq 1$ for all $x$, then
$\Inf[f]\leq d+O\left(\sqrt{dn}e^{-t^{2}/4}\right)$.
\end{lemma}

\begin{proof}
Assume without loss of generality that $p$ is increasing in $\left[-\frac{t}{\sqrt{n}}-\frac{2}{n},
\frac{t}{\sqrt{n}}+\frac{2}{n}\right]$.
We have
\begin{align}
\Inf[f] & = \sum_{i=1}^n \|f_i\|_1= n\EE[|f_1|] \notag \\
 & =n \EE[x_1 f_1(x)]+n \EE[(\sign[f_1(x)]-x_1) f_1(x)] \notag \\
 & \leq \EE[(x_1+\cdots+x_n)f(x)]+n\|(\sign[f_1(x)]-x_1)\|_2 \label{Eq:Aux1}
 \|f_1(x)\|_2,
\end{align}
where the last inequality uses the Cauchy--Schwarz inequality.
We claim that:
\begin{enumerate}
\item $\EE\left[(x_{1}+\cdots+x_{n})f(x)\right]\leq d$,
\item $\|f_1(x)\|_{2} \leq \sqrt{\frac{d}{n}}$, and
\item $\|\sign[f_1(x)]-x_1\|_2 \leq 2\sqrt{2}e^{-\frac{t^{2}}{4}}$.
\end{enumerate}
The first inequality follows from the previous lemma. The second inequality follows from the fact that
$n \EE[f_1(x)^2] = \Inf^{(2)}[f] \leq d$. To see the third inequality,
note that since $p$ is increasing in the interval $\left[-\frac{t}{\sqrt{n}}-\frac{2}{n},
\frac{t}{\sqrt{n}}+\frac{2}{n}\right]$, we have
\[
\left|\frac{\sign [f_1(x)]-x_1}{2} \right| \leq\begin{cases}
1 & \text{if }\left|\frac{x_{1}+\cdots+x_{n}}{n}\right|>\frac{t}{\sqrt{n}},\\
0 & \text{if }\left|\frac{x_{1}+\cdots+x_{n}}{n}\right|\leq\frac{t}{\sqrt{n}}.
\end{cases}
\]
Indeed, when $\left|\frac{x_1+\cdots+x_n}{n}\right| \leq \frac{t}{\sqrt{n}}$ and $x_1=1$, monotonicity of $p$ implies that $f(x) > f(x\oplus e_1)$, and similarly when $x_1=-1$, monotonicity of $p$ implies that $f(x) < f(x\oplus e_1)$.
Therefore, by Hoeffding's inequality,
\[
\|\sign[f_1(x)]-x_1\|_{2}=
2\left\|\frac{\sign[f_1(x)]-x_1}{2}\right\|_2
\leq2\sqrt{\Pr\left[\left|\frac{x_{1}+\cdots+x_{n}}{n}\right|>\frac{t}{\sqrt{n}}\right]}
\leq2\sqrt{2e^{-t^{2}/2}}.
\]
Substituting the three inequalities into~\eqref{Eq:Aux1} yields the assertion of the lemma.
\end{proof}

We are ready now to show our improved upper bound.
\begin{theorem}
Let $f\colon\left\{ -1,1\right\} ^{n}\rightarrow \left[-1,1\right]$ be
a symmetric function of degree $d$. If $n\geq64d^{4}\log d$, then
$\Inf\left[f\right]\leq d+O\left(\sqrt{dn}e^{-t^{2}/4}\right)$, for
$t=\sqrt{n}\left(\frac{1}{4d^{2}}-\frac{2}{n}\right)$.
\end{theorem}

\begin{proof}
At several places in the proof, we assume for convenience that $d$ is large enough; otherwise the theorem is trivial.

Write $f(x)=p\left(\frac{x_{1}+\cdots+x_{n}}{n}\right)$. By Lemma~\ref{lemma:symmetric-new},
$\|p\|_\infty \leq \frac{n}{n-d^2}$. Hence, by Markov--Bernstein's inequality applied to
$p$, we have for all $x \in [-1,1]$,
\[
|p'(x)| \leq d^2 \|p\|_{\infty} \leq d^2 \cdot \frac{n}{n-d^2}.
\]
Applying Markov--Bernstein to $p'$, we obtain for all $x\in\left[-\frac{1}{4d^{2}},
\frac{1}{4d^{2}}\right]$,
\[
|p''(x)| \leq (d^2 \cdot \frac{n}{n-d^2}) \cdot \frac{d-1}{\sqrt{1-x^2}} \leq
(d^2 \cdot \frac{n}{n-d^2}) \cdot \frac{d}{\sqrt{1-1/16d^4}} \leq \frac{3}{2} d^3.
\]
If $p$ is monotone in the interval $x\in\left[-\frac{1}{4d^{2}},
\frac{1}{4d^{2}}\right]$, then the assertion of the theorem follows from the previous
lemma. Otherwise, there exists $x_0 \in\left[-\frac{1}{4d^{2}},\frac{1}{4d^{2}}\right]$
such that $p'(x_0)=0$. By the Mean Value theorem, for any
$y \in \left[-\frac{1}{4d^{2}},\frac{1}{4d^{2}}\right]$ there exists some $z \in \left[-\frac{1}{4d^{2}},\frac{1}{4d^{2}}\right]$ such that
\[
2d^{2}|p'(y)|\leq\left|\frac{p'(y)}{y-x_0}\right| = |p''(z)| \leq \frac{3}{2} d^{3}.
\]
Hence, $p'(y)\leq\frac{3}{4}d$. Now, recall that
$\Inf[f]=\EE[n|f_n(x)|]$.
Since for any $x=(x_1,\ldots,x_n)$,
$n|f_n(x)|=\left|p'(x')\right|$
for some $x'$ in the interval $\left[\frac{x_{1}+\cdots+x_{n-1}}{n}-\frac{1}{n},
\frac{x_{1}+\cdots+x_{n-1}}{n}+\frac{1}{n}\right]$, we have
\[
n|f_n(x)|\leq\begin{cases}
\frac{3}{4}d & \text{if }\left|\frac{x_{1}+\cdots+x_{n-1}}{n}\right|\leq\frac{1}{4d^{2}}-\frac{1}{n},\\
d^{2} & \text{if }\left|\frac{x_{1}+\cdots+x_{n-1}}{n}\right|>\frac{1}{4d^{2}}-\frac{1}{n},
\end{cases}
\]
using Theorem~\ref{thm:d2} in the second case.
Since $n \geq 64d^4\log d$, the event $\left|\frac{x_{1}+\cdots+x_{n-1}}{n}\right|>\frac{1}{4d^{2}}-\frac{1}{n}$
has a negligible probability, and so $\Inf[f]\leq d$.
\end{proof}

When $d \geq n^{1/2}$, we do not know how to improve over the trivial upper bound
$\Inf[f] \leq \sqrt{dn} \leq d^{3/2}$ following from the Cauchy--Schwarz inequality together with the bound $\Inf^{(2)}[f] \leq d$.

\subsection{Upper bound for monotone functions}
\label{sec:sub:monotone}

The upper bound $d/2 \pi +o(d)$ for monotone functions uses a recent result of Klurman~\cite{K12}.
\begin{proposition}[Klurman]\label{Prop:Klurman}
Denote by $X_d$ the set of all degree $d$ univariate polynomials $p\colon[-1,1] \to \RR$
that are monotone. Let $S_d,H_d,F_d \in X_d$ be the following polynomials:
\[
S_d(x)=(1+x) \sum_{i=0}^d (J_i^{(0,1)}(x))^2, \qquad H_d(x)=(1-x^2) \sum_{i=0}^{d-1} (J_i^{(1,1)}(x))^2, \qquad
F_d(x)=\sum_{i=0}^d (J_i^{(0,0)}(x))^2,
\]
where $J_i^{(0,1)}, J_i^{(0,0)}, J_i^{(1,1)}$ are Jacobi polynomials.

For any $d \geq 1$, any $p \in X_d$, and any $x_0 \in [-1,1]$, we have:
\begin{enumerate}
\item $|p'(x_0)| \leq 2 \max(S_k(x_0),S_k(-x_0)) \|p\|_{\infty}$ for $d=2k+2$, and
\item $|p'(x_0)| \leq 2 \max(F_k(x_0),H_k(x_0)) \|p\|_{\infty}$ for $d=2k+1$.
\end{enumerate}
\end{proposition}
Using a classical asymptotic estimate on weighted sums of Jacobi polynomials (see~\cite[Theorem 6.2.35]{N79}), Proposition~\ref{Prop:Klurman} implies
\begin{equation}\label{Eq:Klurman}
|p'(0)| \leq \frac{d}{2\pi} + o(d),
\end{equation}
and the maximum is attained for the polynomial $p$ whose derivative is the corresponding
$S_k$, $F_k$ or $H_k$, depending on the parity of $d$.

\medskip

The reduction from monotone functions on the discrete cube to monotone univariate
polynomials is obtained in two steps. First, we show that one can assume without loss of generality
that the monotone function is symmetric, and then we show that when performing the
reduction described in Section~\ref{sec:sub:symmetric}, the resulting univariate
polynomial can be made as close as we wish to monotone.
\begin{lemma}\label{Lem:Sym-Monotone}
Let $f\colon\{-1,1\}^n \rightarrow \RR$ be monotone, and let $g=\Sym(f)$ be the
symmetrization of $f$. Then $\Inf[g] = \Inf[f]$.
\end{lemma}

\begin{proof}
First, we note that $g$ is monotone. Indeed, for any $\sigma \in S_n$ and any
$x,y \in \{-1,1\}^n$ such that $x_i \leq y_i$ for all $i$, we have $x_{\sigma(i)} \leq y_{\sigma(i)}$
for all $i$. Hence, by the monotonicity of $f$,
\[
f(x_{\sigma(1)},x_{\sigma(2)},\ldots,x_{\sigma(n)}) \leq f(y_{\sigma(1)},y_{\sigma(2)},
\ldots,y_{\sigma(n)}),
\]
and by taking expectation over $\sigma$ we obtain the monotonicity condition for $g$.

It is easy to see that for any monotone function, the total influence is equal
to the sum of the first-level Fourier coefficients. Since both $f$ and $g$ are
monotone, it is thus sufficient to show that
\[
\sum_{i=1}^n \hat f(\{i\})= \sum_{i=1}^n \widehat{\Sym(f)}(\{i\}).
\]
This indeed clearly holds by the definition of symmetrization.
\end{proof}

The lemma implies that there is no loss in generality in considering only symmetric
functions. Moreover, we can assume without loss of generality that $n$ is as large as we wish by using $m$-symmetrization
for a large $m$ instead of symmetrization. (Clearly, the lemma holds without change
for $m$-symmetrization.) The next lemma takes us all the way to Klurman's result
cited above.
\begin{lemma}\label{Lem:Sym-Monotone2}
For any $d \in \NN$, the supremum over the $L_1$ influences of degree $d$ monotone
functions $f\colon\{-1,1\}^n \rightarrow [-1,1]$ is
\[
M_d = \max_{p \in X_d} \frac{p'(0)}{\|p\|_\infty}.
\]
\end{lemma}

\begin{proof}
Let $p \in X_d$ be such that $p'(0) = M_d$ and
$\|p\|_\infty =1$. Define $f^{(n)}\colon\{-1,1\}^n \rightarrow [-1,1]$ by
$f^{(n)}(x_1,\ldots,x_n) = p((x_1+\ldots+x_n)/n)$. By~(\ref{Eq:New-Symmetric1}), as $n$
goes to infinity, $\Inf[f^{(n)}]$ tends to $p'(0)=M_d$.

\medskip

For the other direction, by Lemma~\ref{lemma:symmetric-new}, for any symmetric function
$f\colon\{-1,1\}^n \rightarrow [-1,1]$ with $n > d^2$,
we can write $f(x)=p((x_1+\cdots+x_n)/n)$ for some degree $d$ polynomial $p$ with
$\|p\|_\infty \leq \frac{n}{n-d^2}$. We show now that for $n$ large enough (that can be obtained
by $m$-symmetrization), $p$ can be made as close as we wish to monotone.

By the Markov--Bernstein inequality, $\|p'\|_\infty \leq d^2 \cdot \frac{n}{n-d^2}$.
Applying Markov--Bernstein to $p'$ (which is a degree $d-1$ polynomial on $[-1,1]$), we
obtain$\|p''(x)\|_\infty \leq (d-1)^2 d^2 \cdot \frac{n}{n-d^2}$.
Let $x \in [-1,1]$. Consider the interval $I$ of the form $[-1+2i/n, -1+2(i+1)/n]$ that contains $x$.
As $p$ agrees with $f$ on the endpoints of the interval and $f$ is monotone, there exists $y \in I$
such that $p'(y) \geq 0$. By the Mean Value theorem,
\[
p'(x) \geq p'(y) - (d-1)^2 d^2 \cdot \frac{n}{n-d^2} \cdot \frac{2}{n} \geq - \frac{2d^2(d-1)^2}{n-d^2}.
\]
It follows that the degree $d$ polynomial $\tilde{p}(x) = p(x)+ \frac{2d^2(d-1)^2}{n-d^2} x$
satisfies $\tilde{p} \in X_d$ and $\|\tilde{p}\|_\infty \leq \frac{n+2d^2(d-1)^2}{n-d^2}$. Hence,
\[
p'(0) = \tilde{p}'(0) - \frac{2d^2(d-1)^2}{n-d^2} \leq \frac{n+2d^2(d-1)^2}{n-d^2} M_d - \frac{2d^2(d-1)^2}{n-d^2}.
\]
In particular, for any $\epsilon>0$, for $n$ large enough we have $p'(0) \leq M_d+\epsilon$.
Finally, (\ref{Eq:New-Symmetric1})~implies that as $n$ tends to infinity, $\Inf[f]$ tends
to $p'(0)$. Since $m$-symmetrization allows us to take $n$ as large as we wish, the assertion
follows.
\end{proof}

Combining Lemma~\ref{Lem:Sym-Monotone2} with Proposition~\ref{Prop:Klurman}, we obtain:
\begin{theorem}\label{thm:monotone}
 Let $f\colon \{-1,1\}^n \to [-1,1]$ be a monotone function of degree $d$. Then
 \[
 \Inf[f] \leq \frac{d}{2\pi} + o(d).
 \]
 The maximal influence is attained for the combination of Jacobi polynomials described
 in Proposition~\ref{Prop:Klurman}.
\end{theorem}


A natural question one may ask is, what can be said if the monotone function $f$ is also
Boolean. This appears to be a special case of an open problem, attributed by O'Donnell~\cite{oD12}
to Gopalan and Servedio.
\begin{conjecture}[Gopalan and Servedio, 2009] \label{conj:GopalanServedio}
Let $f\colon\{-1,1\}^n \to \{-1,1\}$. Then $\sum_{i=1}^n \hat f(\{i\}) \leq \sqrt{\deg [f]}$.
\end{conjecture}
In the case of monotone functions, we have $\sum_{i=1}^n \hat f(\{i\})= \Inf[f]$, and thus,
the conjecture asks for an upper bound on the influence in terms of the degree. A recent
result of Scheder and Tan~\cite{ST13} implies the upper bound
$\Inf[f] \leq \ln 2 \cdot \deg[f] (1+o(1))$. Our Theorem~\ref{thm:monotone} yields a slightly
stronger upper bound of $\Inf[f] \leq \frac{1}{2\pi} \cdot \deg[f] (1+o(1))$. However,
this is still very far from the conjectured bound.

{\blue
\subsection{Application to transitive-invariant functions} \label{sec:transitive-invariant}

A function $f\colon \{-1,1\}^n \to \RR$ is called \emph{transitive-invariant} if for every $i,j \in [n]$ there exists a permutation $\sigma \in S_n$ such that $\sigma(i)=j$ and  $f(x_1,\ldots,x_n) = f(x_{\sigma(1)},\ldots,x_{\sigma(n)})$ for every $x=(x_1,\ldots,x_n) \in \{-1,1\}^n$. Note that if $f$ is transitive-invariant, then for every $p$, the influences $\Inf_i^{(p)}[f]$ are all equal.

\begin{proposition}
Every transitive-invariant function $f\colon\{-1,1\}^n\to[-1,1]$ of degree $d$ satisfies, for all $1 \leq p \leq 2$,
\[ \Inf^{(p)}[f] \leq \frac{d^{2p}e^{pd}}{n^{p-1}}. \]
In particular, $\Var[f] \leq \Inf^{(2)}[f] \leq \frac{d^{4}e^{2d}}{n}$.
\end{proposition}
\begin{proof}
Using hypercontractivity (See \cite[Theorem 9.22]{ODonnellBook}), we have
\[
\Inf^{(p)}[f] =\sum_{i=1}^{n}\|f_i\|_{p}^{p} \leq \sum_{i=1}^{n}\|f_i\|_{2}^{p} \leq e^{pd}\sum_{i=1}^{n}\|f_i\|_{1}^{p}=e^{pd}\sum_{i=1}^{n}\left(\Inf_{i}^{(1)}[f]\right)^p
 =e^{pd}\frac{\left(\Inf^{(1)}[f]\right)^p}{n^{p-1}}\leq \frac{d^{2p}e^{pd}}{n^{p-1}}. \qedhere
\]
\end{proof}

This improves on the bound $\Var[f] \leq \frac{e^{O(d)}}{\sqrt{n}}$ proved by Dinur~et~al.~\cite{DFKO}. Since this bound doesn't appear explicitly in~\cite{DFKO}, let us briefly explain how to obtain it from~\cite[Theorem~7]{DFKO}. Putting $J = \emptyset$ and $t = 2$ in the theorem, it states that if $\Var[f] \geq \epsilon$ and $\Inf^{(2)}_i[f] \leq \epsilon^2 C^{-d}/4$ for all $i$ then $\Pr[|f|\geq 2] > 0$, where $C > 0$ is some universal constant. Since $\|f\|_\infty \leq 1$, we deduce that $\Inf^{(2)}[f] = n\Inf^{(2)}_i[f] > n\epsilon^2 C^{-d}/4$, and so $n\epsilon^2 C^{-d}/4 < d$, implying the claimed bound.
}

\section{Tight examples} \label{sec:tightness}

Following Ba\v{c}kurs and Bavarian~\cite{BB}, we conjecture that the total influence of a function $f\colon \{-1,1\}^n \to [-1,1]$ of degree $d$ is at most $d$. In this section we discuss several examples of functions $f$ which achieve or almost achieve this bound.

\paragraph{Boolean homogeneous functions attaining the bound} Proposition~\ref{Prop:Hom-Bool} shows that any function $f$ that is Boolean and homogeneous has total influence exactly $d$.
The quintessential example of such a function is a Fourier character of degree $d$, that is $\chi_S$ for some set $S \subseteq [n]$ of cardinality $|S| = d$. Another example (with $d = 2$) is the function
\[ f_4(x,y,z,w) = \frac{x(z+w) + y(z-w)}{2}. \]
For an arbitrary degree $d \geq 2$, the function $f_4(x_1,x_2,x_3,x_4)x_5\cdots x_{d+2}$ has total influence $d$. This shows that even when $f$ is Boolean and homogeneous, characters are not the unique functions having total influence $d$.

\paragraph{Non-Boolean functions attaining the bound} The following two quadratic functions satisfy $\|f\|_\infty = 1$ and $\Delta f \equiv 2$, and in particular have total influence~$2$:
\begin{align*}
 s(x,y,z,w) &= \frac{xy-zw}{2} + \frac{\sqrt{2}-1}{8}(xz+yw), \\
 t(x,y,z,w) &= \frac{xy-zw}{2} + \frac{\sqrt{2}-1}{16}(x+y)(z+w).
\end{align*}

\paragraph{Symmetric functions almost attaining the bound} Let $p\colon [-1,1] \to [-1,1]$ be a polynomial of degree $d$, and consider the corresponding symmetric function $f(\mathbf{x}) = p(\frac{x_1+\cdots+x_n}{n})$. For large $n$ we have
\[
 \Inf[f] = n\Inf_n[f] \approx n \left| \frac{p(\tfrac{1}{n}) - p(-\tfrac{1}{n})}{2} \right| \approx |p'(0)|.
\]
The Bernstein--Markov theorem (Proposition~\ref{pro:bernstein}) shows that $p'(0) \leq d$. When $d$ is odd, setting $p$ to the Chebyshev polynomial $T_d$ we have $p'(0) = d$, and as $n\to\infty$, the estimates above can be made precise to show that $\Inf[f] \to d$. One could wonder whether these functions provide a counter-example to the conjecture that $\Inf[f] \le \deg[f]$. However, it is not difficult to see that in a deleted neighborhood of $0$, $T_d'(0) < d$, and so for large $n$ the estimates show that  $\Inf[f] < d$, that is, the limit is approached from below. Numerical experiments suggest that $\Inf[f] < d$ holds also for small $n$.

\section{Conjectures} \label{sec:future}

In this section we discuss two directions for improving our results. The first direction aims at improving Theorem~\ref{thm:d2} to a bound of $O(d^{3/2})$ on the total influence. The second direction aims at improving Theorem~\ref{thm:dlogd} to a bound of $O(d)$ on the total influence of homogeneous functions.

\subsection{General functions} \label{sec:general}

We start by proving an $O(d^{3/2})$ bound on the total influence of homogeneous functions. While Theorem~\ref{thm:dlogd} provides a better upper bound of $O(d\log d)$, this new method could potentially extend to general functions. The proof uses Sarantopoulos's extension (Proposition~\ref{pro:sar}) of the Markov--Bernstein theorem. (We could also use the classical Bernstein's theorem.)


\begin{theorem} \label{thm:d32}
 Let $f\colon \{-1,1\}^n \to [-1,1]$ be a homogeneous function of degree $d$. Then
 \[ \Inf[f] \leq \|\Delta(f)\|_\infty \leq O(d^{3/2}). \]
\end{theorem}
\begin{proof}
 Let $\alpha = 1 - 1/d$, and define $g = T_\alpha f$. Note that $g(\mathbf{x}) = f(\alpha \mathbf{x})$ and similarly $\Delta(g)(\mathbf{x}) = \Delta(f)(\alpha\mathbf{x})$. Since $|\alpha| \leq 1$, the interpretation of $T_\alpha f$ as an averaging operator shows that $\|g\|_\infty \leq \|f\|_\infty \leq 1$. As in the proof of Theorem~\ref{thm:d2}, Proposition~\ref{pro:sar} shows that for all $\mathbf{x} \in \{-1,1\}^n$,
 \[ \Delta(g)(\mathbf{x}) = \Delta(f)(\alpha\mathbf{x}) \leq \frac{d}{\sqrt{1-\alpha^2}}. \]
 Since $g$ is homogeneous, $\Delta(g)(\mathbf{x}) = \alpha^d \Delta(f)(\mathbf{x})$, and so
 \[ \Delta(f)(\mathbf{x}) \leq \alpha^{-d} \frac{d}{\sqrt{1-\alpha^2}} = O(d^{3/2}). \qedhere \]
\end{proof}

When $f$ is not homogeneous, we can try to fix the argument as follows.
\begin{lemma} \label{lem:d32}
 Let $f\colon \{-1,1\}^n \to [-1,1]$ be a function of degree $d$. For all $\alpha \in [-1,1]$ we have
\[ \Inf[f] \leq \max_{i \in [n]} \frac{\|f_i\|_1}{\|T_\alpha f_i\|_1} \frac{d}{\sqrt{1-\alpha^2}}. \]
\end{lemma}
\begin{proof}
 Fix $\alpha$, and let $g = T_\alpha f$. As in Theorem~\ref{thm:d32},
 \[ \|\Delta(g)\|_\infty \leq \frac{d}{\sqrt{1-\alpha^2}}. \]
 On the other hand, as $g_i=T_\alpha f_i$,
\[
 \Inf[f] = \sum_{i=1}^n \|f_i\|_1 = \sum_{i=1}^n \frac{\|f_i\|_1}{\|T_\alpha f_i\|_1} \|g_i\|_1 \leq \left( \max_{i \in [n]} \frac{\|f_i\|_1}{\|T_\alpha f_i\|_1} \right) \Inf[g] \leq \left(\max_{i \in [n]} \frac{\|f_i\|_1}{\|T_\alpha f_i\|_1}\right) \frac{d}{\sqrt{1-\alpha^2}}. \qedhere
\]
\end{proof}

This prompts the following definition.
\begin{definition} \label{def:const}
 Let $d \geq 1$ and $\alpha \in [-1,1]$. Define
\[
 C_{d,\alpha} = \sup_f \frac{\|f\|_1}{\|T_\alpha f\|_1},
\]
 where the supremum ranges over all $n$ and all functions $f\colon \{-1,1\}^n \to \RR$ of degree at most $d$.
\end{definition}

We can restate the conclusion in Lemma~\ref{lem:d32} as follows:
\[ \Inf[f] \leq C_{d,\alpha} \frac{d}{\sqrt{1-\alpha^2}}. \]
In particular, if $C_{d,1-1/d} = O(1)$ then $\Inf[f] = O(d^{3/2})$.

The best bound on $C_{d,\alpha}$ we can prove is the following.

\begin{lemma} \label{lem:Cdalpha}
 For all functions $f\colon \{-1,1\}^n \to \RR$ of degree $d$ and all $\alpha \in (0,1]$,
\[
 \|T_\alpha f\|_1 \geq \alpha^{\min(d^2,n)} \|f\|_1.
\]
 In particular, $C_{d,\alpha} \leq \alpha^{-d^2}$.
\end{lemma}
\begin{proof}
 We start by showing that $\|T_\alpha f\|_1 \geq \alpha^n \|f\|_1$. Note first that for all $i \in [n]$, we have $\|f_i\|_1 \leq \|f\|_1$. This follows from
 \[
  |f_i(x)| + |f_i(x \oplus e_i)| = 2\left|\frac{f(x) - f(x\oplus e_i)}{2}\right| \leq |f(x)| + |f(x \oplus e_i)|.
 \]
 Now, it is well-known that
\[
 \left. \frac{dT_{e^{-\epsilon}} f(\mathbf{x})}{d\epsilon} \right|_{\epsilon=0} = -Lf(\mathbf{x}),
\]
 where the Laplacian $Lf$ is given by $Lf = f_1 + \cdots + f_n$. Therefore
\[
 \left. \frac{d\|T_{e^{-\epsilon}} f\|_1}{d\epsilon}  \right|_{\epsilon=0}  \geq -\|Lf\|_1 \geq -\sum_{i=1}^n \|f_i\|_1 \geq -n\|f\|_1.
\]
 Let $\phi(\delta) = \|T_{e^{-\delta}} f\|_1$. Applying the inequality above to $T_{e^{-\delta}} f$ shows that $\phi'(\delta) \geq -n\phi(\delta)$ and so $(\log \phi(\delta))' \geq -n$. Integrating, we obtain $\phi(\delta)/\phi(0) \geq e^{-\delta n}$. Taking $\delta = -\log \alpha$, we deduce
\[
 \|T_\alpha f\|_1 = \phi(\delta) \geq e^{-\delta n} \phi(0) = \alpha^n \|f\|_1.
\]

\medskip

 We proceed with the proof that $\|T_\alpha f\|_1 \geq \alpha^{d^2} \|f\|_1$ (we thank K.~Oleszkiewicz for help with this proof). Let $V_d$ denote the vector space of all real-valued polynomials $P$ of degree at most~$d$ satisfying $|P(x)| \leq 1$ for all $|x| \leq 1$. Standard results in functional analysis (see for example Rivlin's book~\cite{Rivlin}) show that every linear functional $\Phi\colon V_d \to \RR$ can be represented as
 \[ \Phi \colon  P \mapsto \sum_i c_i P(\alpha_i) \]
 for some points $\alpha_i \in [-1,1]$, in such a way that the maximum of $\Phi P$ over $V_d$ is $\sum_i |c_i|$. Applying this result to the functional $\Phi f\colon f \mapsto f'(1)$, we obtain coefficients $c_i,\alpha_i$ satisfying $\sum_i |c_i| = d^2$ (according to Markov's inequality). Since $\Phi$ maps $x^k$ to $k$ for all $k \leq d$, for such $k$ we have
 \[ \sum_i c_i \alpha_i^k = k. \]
 Since $f$ has degree $d$, this implies that
\[
 \sum_i c_i T_{\alpha_i} f = \sum_{S \subseteq [n]} \hat{f}(S) \chi_S \sum_i c_i \alpha_i^{|S|} = \sum_{S \subseteq [n]} |S| \hat{f}(S) \chi_S = Lf.
\]
 The interpretation of $T_\beta$ as an average shows that $\|T_\beta f\|_1 \leq \|f\|_1 \leq 1$ for every $\beta \in [-1,1]$. In particular,
 \[ \|Lf\|_1 \leq \sum_i \|c_i T_{\alpha_i} f\|_1 \leq \sum_i |c_i| \|f\|_1 = d^2 \|f\|_1. \]
 As in the preceding half of the proof, this implies that $\|T_\alpha f\|_1 \geq \alpha^{d^2} \|f\|_1$.
\end{proof}

Unfortunately, plugging this bound in Lemma~\ref{lem:d32} does not result in any improvement over Theorem~\ref{thm:d2}.

\subsection{Homogeneous functions} \label{sec:homogeneous}

Theorem~\ref{thm:dlogd} shows that the total influence of a homogeneous function of degree $d$ is at most $O(d\log d)$. The argument relies on a result of Harris~\cite{H2} showing that a real polynomial satisfying $|h(\epsilon)| \leq (1+\nobreak|\epsilon|)^d$ for all $\epsilon \in \RR$ also satisfies $|h'(0)| = O(d\log d)$; R\'ev\'esz and Sarantopoulos~\cite{RS} show that the bound on $|h'(0)|$ is tight. Recall that the function $h(\epsilon)$ figures in the proof in the following way. For a certain set $S \subseteq [n]$, we define $g(x,y) = f(\overbrace{x,\ldots,x}^{S},\overbrace{y,\ldots,y}^{\comp{S}})$. Since $f$ is multilinear and homogeneous, $|g(x,y)| \leq \max(|x|,|y|)^d$, and so the function $h(\epsilon) = g(1+\epsilon,1-\epsilon)$ satisfies $|h(\epsilon)| \leq (1+|\epsilon|)^d$.

R\'ev\'esz and Sarantopoulos comment that every real polynomial $h$ satisfying $|h(\epsilon)| \leq (1+|\epsilon|)^d$ can be lifted to a bivariate homogeneous polynomial $g(x,y)$ given by $g(x,y) = y^d h(x/y)$. This polynomial satisfies $|g(x,y)| \leq |y|^d(1 + |x|/|y|)^d = (|x| + |y|)^d$, and so the bound on the derived $|h'(0)|$ can be achieved by some function $g(x,y)$ satisfying $|g(x,y)| \leq (|x| + |y|)^d$. In our case, we have the stronger guarantee $|g(x,y)| \leq \max(|x|,|y|)^d$. We can modify the proof to reflect this stronger guarantee.

\begin{definition} \label{def:homogeneous}
 Let $K_d$ be the supremum of $|h'(1)|$ over all polynomials satisfying $|h(\epsilon)| \leq \max(1,|\epsilon|^d)$ for all $\epsilon \in \RR$.
\end{definition}

\begin{lemma} \label{thm:Kd}
 Let $f\colon \{-1,1\}^n \to [-1,1]$ be a homogeneous function of degree $d$. Then
 \[ \Inf[f] \leq \|\Delta(f)\|_\infty \leq 2K_d. \]
\end{lemma}
\begin{proof}
 Obviously $\Inf[f] = \|\Delta(f)\|_1 \leq \|\Delta(f)\|_\infty$. It remains to verify that $|\Delta(f)(x)| \le 2K_d$ for every $x \in \{-1,1\}^n$. Note that without loss of generality we only need to prove this for $x=\mathbf{1}$. Let $S$ be the set of $i \in [n]$ such that $f_i(\mathbf{1}) \geq 0$. Then
 \[
  |\Delta(f)(\mathbf{1})| = \sum_{i \in S} f_i(\mathbf{1}) - \sum_{i \in \comp{S}} f_i(\mathbf{1}).
 \]
 Define the bivariate polynomial $g(x,y) = f(\overbrace{x,\ldots,x}^{S},\overbrace{y,\ldots,y}^{\comp{S}})$. Since $f$ is multilinear and homogeneous, $|g(x,y)| \leq \max(|x|,|y|)^d$. In particular, the functions $h_1(\epsilon) = g(\epsilon,1)$ and $h_2(\epsilon) = g(1,\epsilon)$ are polynomials satisfying $|h_i(\epsilon)| \leq \max(1,|\epsilon|)^d$ for $i=1,2$. By definition, $|h_i'(1)| \leq K_d$ for $i=1,2$. Thus
\[
|\Delta(f)(\mathbf{1})| = \sum_{i \in S} f_i(\mathbf{1}) - \sum_{i \in \comp{S}} f_i(\mathbf{1}) =
\sum_{i \in S} \frac{\partial f}{\partial x_i} (\mathbf{1}) - \sum_{i \in \bar{S}}
\frac{\partial f}{\partial x_i}(\mathbf{1})  =h'_1(1)-h'_2(1)\leq 2K_d. \qedhere
\]
\end{proof}

Harris~\cite{H3,H4} develops systematically a method aimed toward computing constants like $K_d$ using Lagrange interpolation. R\'ev\'esz and Sarantopoulos~\cite{RS} present a different framework which employs potential theory. We believe that these methods can be used to estimate $K_d$ asymptotically. We conjecture that $K_d = \Theta(d)$, leading to a proof that $\Inf[f] \leq O(d)$ for homogeneous functions.

We have computed $K_1 = 1$ and $K_2 = 1 + \sqrt{2}$. The bound for $d = 1$ is attained for $h(x) = \pm x$. For $d = 2$, it is attained for
\[ h(x) = \left(\frac{1}{2} + \frac{1}{2\sqrt{2}}\right) (x^2 - 1) + \frac{1}{\sqrt{2}} x. \]
The upper bound $K_1 \leq 1$ is trivial. The upper bound $K_2 \leq 1 + \sqrt{2}$ follows by Lagrange interpolation (following Harris) with the points $1 \pm \sqrt{2}$ (a priori, the method requires three points, but the third point cancels out in the calculation).

\bibliographystyle{alpha}
\bibliography{L1}

\end{document}